\documentclass[reqno,11pt]{amsart}
\usepackage{amsmath, latexsym, amsfonts, amssymb, amsthm, amscd}
\usepackage[utf8]{inputenc}
\usepackage{graphics,epsf,psfrag,epsfig}

\numberwithin{equation}{section}

\newtheorem{theorem}{Theorem}[section]

\newtheorem{proposition}[theorem]{Proposition}
\newtheorem{cor}[theorem]{Corollary}
\newtheorem{rem}[theorem]{Remark}

\renewcommand{\ge}{\geq}
\renewcommand{\le}{\leq}
\newcommand{\ind}{\mathbf{1}}

\newcommand{\R}{\mathbb{R}}
\newcommand{\Z}{\mathbb{Z}}
\newcommand{\N}{\mathbb{N}}
\renewcommand{\tilde}{\widetilde}
\renewcommand{\hat}{\widehat}

\newcommand{\cC}{{\ensuremath{\mathcal C}} }

\newcommand{\cT}{{\ensuremath{\mathcal T}} }

\newcommand{\bP}{{\ensuremath{\mathbf P}} }
\newcommand{\bE}{{\ensuremath{\mathbf E}} }

\DeclareMathSymbol{\leqslant}{\mathalpha}{AMSa}{"36} 
\DeclareMathSymbol{\geqslant}{\mathalpha}{AMSa}{"3E} 
\DeclareMathSymbol{\eset}{\mathalpha}{AMSb}{"3F}     
\renewcommand{\leq}{\;\leqslant\;}                   
\renewcommand{\geq}{\;\geqslant\;}                   

\newcommand{\bbR}{{\ensuremath{\mathbb R}} }

\newcommand{\bbZ}{{\ensuremath{\mathbb Z}} }

\newcommand{\gep}{\varepsilon}       

\newcommand{\gG}{\Gamma}

\newcommand{\gO}{\Omega}

\newcommand{\gL}{\Lambda}

\makeatletter
\def\captionfont@{\footnotesize}
\def\captionheadfont@{\scshape}

\long\def\@makecaption#1#2{%
  \vspace{2mm}
  \setbox\@tempboxa\vbox{\color@setgroup
    \advance\hsize-6pc\noindent
    \captionfont@\captionheadfont@#1\@xp\@ifnotempty\@xp
        {\@cdr#2\@nil}{.\captionfont@\upshape\enspace#2}%
    \unskip\kern-6pc\par
    \global\setbox\@ne\lastbox\color@endgroup}%
  \ifhbox\@ne 
    \setbox\@ne\hbox{\unhbox\@ne\unskip\unskip\unpenalty\unkern}%
  \fi
  \ifdim\wd\@tempboxa=\z@ 
    \setbox\@ne\hbox to\columnwidth{\hss\kern-6pc\box\@ne\hss}%
  \else 
    \setbox\@ne\vbox{\unvbox\@tempboxa\parskip\z@skip
        \noindent\unhbox\@ne\advance\hsize-6pc\par}%
\fi
  \ifnum\@tempcnta<64 
    \addvspace\abovecaptionskip
    \moveright 3pc\box\@ne
  \else 
    \moveright 3pc\box\@ne
    \nobreak
    \vskip\belowcaptionskip
  \fi
\relax
}
\makeatother
\def\writefig#1 #2 #3 {\rlap{\kern #1 truecm
\raise #2 truecm \hbox{#3}}}

\newcommand{\Tm}{T_{\rm mix}}

\author{Hubert Lacoin}
\address{CEREMADE, Place du Mar\'echal De Lattre De Tassigny
75775 PARIS CEDEX 16 - FRANCE}
\email{lacoin@ceremade.dauphine.fr}
\title[Zero-temperature stochastic Ising model in any dimension]{Approximate Lifshitz law for the
zero-temperature stochastic Ising  model in
any dimension}

\begin{document}

\begin{abstract}
We study the Glauber dynamics for the zero-temperature stochastic Ising model in dimension 
$d\ge 4$ with ``plus'' boundary condition.
Let $\cT_+$ be the time needed for an hypercube of size $L$ 
entirely filled with ``minus'' spins to become entirely ``plus''. We prove that $\cT_+$  
is $O(L^2(\log L)^c)$ for some constant $c$, 
not depending on the dimension. This brings further rigorous justification for the so-called ``Lifshitz law''
 $\cT_+=O(L^2)$ \cite{Lif, FH} conjectured on heuristic grounds. 
The key point of our proof is to use the detailed knowledge that we have on the three-dimensional problem: 
results for fluctuation 
of monotone interfaces at equilibrium and mixing time for monotone interfaces dynamics extracted from \cite{CMST}, 
to get the result in higher dimension.
\end{abstract}

\maketitle 

\section{Introduction}

We study the zero-temperature stochastic Ising model in a finite hypercube of side-length $L$ in $\Z^d$ with plus boundary condition.
Initially, the whole cube is filled with minus spins. The spins evolve following the majority rule:
with rate one each spin takes the same value as the majority of
its neighbors (including those at the boundary) if the latter is well
defined. Otherwise its sign is determined by tossing a fair coin.  
Eventually, all minus spins disappear, as a result of the pressure imposed by plus spins from the boundary.

Our aim is to study the time $\cT_+$ needed for the spins in the hypercube to become entirely plus.
 According to a heuristics from Lifshitz \cite{Lif}, on a macroscopic
 scale, each point of the rescaled interface 
between plus and minus spins should move feeling a local drift
proportional to its local mean curvature (this conjecture has been recently proved with co-authors for the 
zero-temperature two-dimensional case with macroscopically convex initial domain of $-$ spins:
see \cite{LST}). This readily implies that,
with high probability, $\cT_+=O(L^2)$.

Here we prove that, with high probability, $\cT_+=O(L^2(\log L)^{10})$
for all dimension $d\ge 4$. Our result complements existing analogous bounds for dimension 
$d=2$ and $d=3$ obtained in \cite{FSS, CMST}. 
It is important to keep in mind that the cases $d=2$ and $d=3$ are completely
different, mainly because the equilibrium fluctuations of Ising
interfaces in $d=2$ and in $d=3$ occur on very different scales ($O(L^{1/2})$ for $d=2$ and
$O(\log(L))$ for $d=3$).  As such, they have been analyzed by using very
different approaches. The case $d\ge 4$ should be more similar to the
case $d=3$, which therefore plays the role of  a critical
dimension. However, contrary to what has been done
in \cite{CMST}, our proof does not attempt to control the local mean drift of the
interface. Such an approach, in fact, would at least require two main missing tools: (i) a detailed analysis
of the local equilibrium fluctuations of Ising interfaces\footnote{In dimension $d\ge 4$ the
  fluctuations of Ising hyper-surfaces are believed to be order
  one. Recently a result of this sort has been proved under a Lifshitz
  condition \cite{Peled}.}; (ii) good
estimates on mixing times for $(d-1)$-monotone surfaces, following
e.g.\ the method developed by Wilson \cite{Wilson} for the three-dimensional case.  Instead, we
use an ad-hoc construction that allows us to bound the hitting time
$\cT_+$ for the $d$-dimensional dynamics using known sharp estimates for the
same hitting time in three dimensions \cite{CMST}.
One of the reasons why this approach is successful is that in dimension three, Ising interfaces
are flat with only logarithmic fluctuations
(see e.g.\ Proposition 4 in \cite{CMST} for more details). 

\medskip

Lifshitz's law is a general prediction for the mixing-time of the stochastic Ising model in the whole low-temperature
phase and not only for the zero-temperature version presented here. However, rigorous results available at positive temperature 
(especially concerning the upper-bound) are still very far 
from it. We refer to \cite{LMST} for the best available result in dimension $2$ and to its bibliography for a
full account on what is known at positive temperature. Getting polynomial upper-bound (in $L$)on the mixing-time in dimension $2$ and even
sub-exponential bounds in dimension $3$ are challenging open problems.
Note also that our result can be extended to some cases where the temperature tends to zero as the size of the system $L$ goes to infinity, 
as in \cite{CMST}.

\section{Model and result}

\subsection{Definition of the model}\label{pbrsut}

Given a finite subset $\gG\subset \bbZ^d$ we study a continuous time Markov chain on the state space
\begin{equation}
 \gO_{\gG}:=\{-1,+1\}^{\gG}.
\end{equation}
This Markov chain can be seen as heat-bath Glauber dynamics for the Ising model at zero temperature.
We write $\sigma=(\sigma_x)_{x\in \gG}$ for a generic element of $\gO_{\gG}$.
We refer to $\sigma_x=\pm 1$ as \textit{spins}, and will often write simply $+$ or $-$. 
Let 
\begin{equation}
 \partial \gG:=\{y\in \bbZ^d \setminus \gG \ | \ \exists x \in \gG, x \sim y\}
\end{equation}
denote the boundary (in $\bbZ^d$) of $\gG$ and fix
some $\eta \in \{+1,-1\}^{\partial \gG}$ (referred to as boundary condition).
We define also the internal boundary of a set to be
\begin{equation}
  \partial^- \gG:=\{y\in \gG  \ | \ \exists x \in \Z^d \setminus\gG, x \sim y\}.
\end{equation}

Given $\xi \in \gO_{\gG}$, we consider the dynamics    $\left(\sigma^{\xi}(t)\right)_{t\ge 0}$ starting 
from initial condition $\sigma^{\xi}(0)=\xi$ (one may also write $\sigma^{\xi,\eta}(t)$ when one wants to underline dependence with 
respect to the boundary condition) and with the following evolution rules:
sites $x\in \gG$ are equipped with independent Poisson clocks with rate $1$, and the spin at $x$ may change its 
value only when the clock at $x$ rings. If the clock at $x$ rings at time $t_0$,
then one looks at the spins of the $2d$ neighbors of $x$  (some of which may lie in $\partial \gG$, in that case their spin is given by  
$\eta$) just before $t_0$, and then
\begin{itemize}
 \item if a strict majority of neighboring spins are $+$ then set $\sigma_{x}(t_0)=+$,
 \item if a strict majority of neighboring spins are $-$ then set $\sigma_{x}(t_0)=-$,
 \item if there are exactly $d$ pluses and $d$ minuses in the neighborhood of $x$ then set $\sigma_{t_0}(x)=+$ or $-$ with probability 
$1/2$ independently of the rest of the process.
\end{itemize}

\begin{rem}\rm \label{constru}
 One gets a natural construction of the dynamics by equipping every $x\in \gG$ with independent Poisson 
clock processes $(\tau_{n,x})_{n\ge 0}$ 
of rate $1$  and with independent sequences $(X_{n,x})_{n\ge 0}$ of  i.i.d. Bernoulli variables of parameter $1/2$ 
(taking values in $\{-1,+1\}$).
Then the evolution is set as follows:

At time $\tau_{n,x}$, if a strict majority of the neighbors of $x$ have spin $+$ (resp.\ $-$) one sets $\sigma_x=+$, (resp.\ $-$), 
otherwise one sets $\sigma_x(\tau_{n,x})=X_{n,x}$.

This construction (sometimes referred to as the \textit{graphical construction} in the literature) is used in the proof, 
and it is very useful to provide coupling for stochastic comparisons.

\end{rem}

More formally, one can write explicitly the transition rates $p$ for our Markov chain as follows.
Given $\sigma\in \gO_{\gG}$ and $x\in \gG$ let $\sigma^{x,+}$,(resp. $\sigma^{x,-}$) 
denote the configuration which is identical to $\sigma$ on 
$\gO\setminus \{x\}$ and where $\sigma^{x,+}_x=+$ (resp. $\sigma^{x,-}_x=-$). Then
\begin{equation}
 p(\sigma^{x,+},\sigma^{x,-})=\left\{
    \begin{array}{lll}
       1   & \text{ if }  \sum_{y\sim x,\ y\in \gG} \sigma_y + \sum_{y\sim x,\ y\in \partial\gG} \eta_y<0,\\  
       1/2 & \text{ if }  \sum_{y\sim x,\ y\in \gG} \sigma_y + \sum_{y\sim x,\ y\in \partial\gG} \eta_y=0,\\
       0   & \text{ if }  \sum_{y\sim x,\ y\in \gG} \sigma_y + \sum_{y\sim x,\ y\in \partial\gG} \eta_y>0.
    \end{array}\right.
\end{equation}
\begin{equation}
 p(\sigma^{x,-},\sigma^{x,+})=\left\{
    \begin{array}{lll}
       0   & \text{ if }  \sum_{y\sim x,\ y\in \gG} \sigma_y + \sum_{y\sim x,\ y\in \partial\gG} \eta_y<0,\\  
       1/2 & \text{ if }  \sum_{y\sim x,\ y\in \gG} \sigma_y + \sum_{y\sim x,\ y\in \partial\gG} \eta_y=0,\\
       1   & \text{ if }  \sum_{y\sim x,\ y\in \gG} \sigma_y + \sum_{y\sim x,\ y\in \partial\gG} \eta_y>0.
    \end{array}\right.
\end{equation}
All other transitions have rate zero.  We denote by $\bP$ 
the probability associated with this Markov chain
(although boundary condition and domain considered may vary) and by $\bE$ the corresponding expectation.
We are interested more specifically in the case where $\gG=\gG_L:= \{1,\dots, L\}^d$ is the $d-$dimensional hypercube, 
and $\eta$ is uniformly 
$+$ on $\partial \gG_L$. 
In that case, the dynamics possesses a unique absorbing state: the configuration with $+$ spins everywhere on $\gG_L$ 
(for simplicity we write it $+$).
We are interested in estimating the hitting time of this absorbing state starting from the all $-$ configuration

\begin{equation}
\mathcal \cT_{+}:= \inf\{t\ge 0 \ | \ \sigma^{-}(t)= +\}.
\end{equation}
In fact our main result is an upper bound on the mixing time defined here as
\begin{equation}
\Tm:=  \inf\{t\ge 0 \ |\ \bP[\mathcal T_+>t]\le 1/4\}.
\end{equation}

It has been conjectured that the growth of $\Tm$ should be given by Lifshitz's law: $\Tm$ should be of order $L^2$  
(see \cite{Lif} and \cite{FH} where a more general formulation of this conjecture can be found).
This result has already been proved in dimension $2$ (exactly) and $3$ 
(with non-matching logarithmic corrections for the lower and upper bound) 
by Fontes \textit{et al.} \cite{FSS} (exact upper bound in dimension $2$) by Caputo \textit{et al.} 
\cite{CMST} (lower bound in dimension $2$ and result in dimension $3$). In 
\cite{LST}, the first order asymptotic of $\Tm$ ($\Tm=L^2/2(1+o(1))$) is given for $d=2$.
We refer to \cite{CMST} and its bibliography for a more complete 
introduction to this problem and underlying issues.
The aim of this note is to extend this result by showing that there exists some constant $c_0>0$
independent of the dimension $d$ such that
\begin{equation}
\Tm=O(L^2 (\log L)^{c_0}).
\end{equation}
Previously in \cite[Theorem 1.3]{FSS}, the bound $\Tm= O(L^d)$ had been derived from an iterative procedure using as a starting point 
the result for 
$d=2$. The same kind of procedure starting from the result of \cite{CMST} would give $\Tm= 0(L^{d-1}(\log L)^c)$
which can be
considered as the best bound that had been known so far (see Proposition 
\ref{cuicuicui} for the statement and its proof for $d=4$).

\medskip

Our proof strongly relies on estimates that have been obtained for the three-dimensional model \cite{CMST}.
Sadly enough, we cannot yet give a lower bound on $\Tm$ with the same polynomial order. The reason for this
is that methods relying on exact computation for $d=2$ and $3$ developed in a paper of Wilson \cite{Wilson} and used in \cite{CMST}
fail to work in dimension larger $4$.

\medskip

We state now our main result in full detail.

\begin{theorem}\label{mainres}
Consider the dynamics in $\gG_L:=\{1,\dots,L\}^d$, with boundary condition $+$ on $\partial \gG_L$. 
There exists $c_0$ such that for all $d\ge 4$,  the following holds for every $k >0$. Uniformly for all initial configuration $\xi$, 
\begin{equation}\label{leres}
\bP\left[ \exists x\in \gG_L,\ \sigma_x^{\xi}(L^2(\log L)^{c_0})=-  \right]= O(L^{-k}).
\end{equation}
\end{theorem}

We stress once again that the constant $c_0$ in the result does not depend on the dimension.
In fact, examining the proof in \cite{CMST} a (non-optimal) choice can be to choose any $c_0> 19/2$. 
In what follows, we fix $c_2:=3/2$, $c_1>13/2$ and $c_0:=c_1+ 2c_2+\gep$ (for some small $\gep$).

\medskip

\begin{rem}[about the lower-bound]\rm
The best lower-bound available for $\mathcal T_+$ in dimension $d\ge 4$ is rather trivial:  
there exists a constant $c_d$ such that, starting from the all $-$ configuration,
\begin{equation}
 \lim_{L\to \infty}\bP\left[\mathcal T_+\ge c_d L\right]=1.
\end{equation}
To obtain such a bound, it is sufficient to consider that when $\mathcal T_+$ occurs, the center of the cube has turned from $-$ to $+$ so that
there has been a 
sequence of consecutive updates on a paths of nearest neighbor sites, starting from the boundary of the cube and
ending in its center. 
Any improvement on this bound would be of interest.
\end{rem}

\section{Preparatory work}\label{prepw}

\subsection{Monotonicity of the dynamics}

We introduce a partial order on $\gO_{\gG}$ (resp.\ $\gO_{\partial \gG}$), saying that $\sigma\le \sigma'$ if $\sigma_x\le \sigma'_x$ for every $x\in \gG$
(resp.\ $x\in\partial \gG$).
Our dynamics enjoys nice monotonicity properties with respect to this order.
Let $\mu_{t,\eta}^{\xi}$ denote the law of $\sigma^{\xi}(t)$ with boundary condition $\eta$.
Then if $\eta \le \eta'$ and $\xi\le \xi'$ 
\begin{equation}
\mu_{t,\eta}^{\xi}\preceq \mu_{t,\eta'}^{\xi'},
\end{equation}
where $\preceq$ denotes stochastic domination (i.e.\ $\mu\preceq \nu$ if for any increasing function 
$f:\gO_{\gG}\rightarrow \R$,
$\mu(f)\le \nu(f)$, a function $f$ being increasing if $f(\sigma)\le f(\sigma')$ for all $\sigma\le \sigma'$).
This can be proved by constructing a coupling of the two processes $\sigma^{\xi,\eta}(t)$ and $\sigma^{\xi',\eta'}(t)$, 
building the two processes as explained in Remark \ref{constru}, with the  same realization of $\tau_{n,x}$ and $X_{n,x}$ for both.
With this coupling one has that for every $t\ge 0$
\begin{equation}\label{domin}
  \sigma^{\xi,\eta}(t)\le \sigma^{\xi',\eta'}(t).
\end{equation}
 
\medskip
 
This result remains valid with boundary conditions that vary through time:
if $\eta(\cdot)$ and $\eta'(\cdot)$ are (deterministic or random) functions from $\bbR_+$ to $\gO_\gG$ such that 
$\eta(s)\le \eta'(s)$ for all $s\in [0,t]$, one can still construct the Markov chains $\sigma^{\xi,\eta(\cdot)}(t)$, $\sigma^{\xi,\eta(\cdot)}(t)$  
(in this case the transition rates are time dependent), and couple them in a way that \eqref{domin} holds.

\medskip

A simple consequence of this monotonicity property is that to prove our main theorem, 
we just need to prove it with initial configuration all $-$. This monotonicity of the dynamics is crucial for our proof as it allows us to 
use an explicit strategy to control the evolution of the set of $-$ (see Proposition \ref{maintool} below).

\subsection{Main line of the proof}

From now on, we focus on the case $d=4$.
We show at the end of the paper how to adapt the proof for arbitrary dimension $d$.

\medskip

The first tool of the proof is to perform surgical changes on the original dynamics that cancels some of 
the spin-flips 
from $-$ to $+$. Due to the above mentioned monotonicity properties, this slows down the dynamics and makes the mixing time longer 
(which is alright to get an upper-bound), but if 
these changes are done carefully, this can also guide the dynamic in a nice pattern which allows 
us to have a better control on it.

\medskip

We choose carefully our modifications so that the dynamics we obtained is mostly a product of three-dimensional dynamics.
We consider our four-dimensional cube of side-length $L$ (we will rather consider a cylinder in what follows), as
a superposition of $L$ three dimensional \textit{slices}, and use our update blocking strategy to make the evolution of 
these slices independent. Then we use some ingredients borrowed from \cite{CMST} to control the evolution of each slice.
 In order to do this without slowing down the dynamics too much one needs an \textit{ad-hoc} geometrical construction 
that we present at the end of this section. In order to make the reader more aware of this need, 
one exposes first a naive and non-optimal strategy that 
yields consequently a non-optimal bound.

\subsection{A non-optimal strategy}

A simple way to control the dynamics in to make the successive three-dimensional layers of our four-dimensional cube evolve one after 
the other. This gives the following result,

\begin{proposition}\label{cuicuicui}
When $d=4$, there exists a constant $c$ such that for every $k\ge 0$, 
\begin{equation}
 \bP[\mathcal T_+> L^3 (\log L)^c]=O(L^{-k})
\end{equation}

\end{proposition}

\begin{proof}
One consider $\tilde \sigma$ a modified version of the dynamics $\sigma^-$ (starting from all $-$),
that blocks the update of spins in the $i$-th layer $\{1,\dots, L\}^3\times \{i\}$ up to time 
$(i-1) L^2 (\log L)^c$. By monotonicity $\tilde \sigma$ is dominated by $\sigma$ and thus it is sufficient 
to prove the proposition
with $\sigma$ replaced by $\tilde \sigma$.
\medskip

We prove by induction on $i=1,\dots,L$ that for $L$ large enough,
\begin{equation}\label{dertu}
 \bP\left[\exists x\in \{1,\dots, L\}^3\times \{1,\dots, i\}, \tilde\sigma_x (iL^2(\log L)^c)=-\right]\le iL^{-k-1}.
\end{equation}
For $i=1$, notice that only the sites  in the first layer, $\{1,\dots, L\}^3\times \{1\}$ can be updated 
during the time window 
$[0,L^2 (\log L)^c]$. For each $x$ in this layer, spins of $x+e_4$ and $x-e_4$ (where $e_4:=(0,0,0,1)$) are respectively 
$-$ and $+$ so 
that neighbors along the fourth dimension have no influence on the majority rule. As a result the dynamics in 
$\{1,\dots, L\}^3\times \{1\}$ is in bijection with the three-dimensional dynamics on a cube 
with $+$ boundary condition and by \cite[Theorem 1]{CMST} (strictly speaking the result says 
that the probability that a $-$ spin remains in the 
cube after time $L^2 (\log L)^c$ is $O(L^{-1})$ but from the proof it can be shown that 
it is $O(L^{-k})$ for all $k$),
\eqref{dertu} holds with $d=1$.

\medskip

For the induction step, notice that if all spins of $\{1,\dots, L\}^3\times \{1,\dots, i\}$ are equal to $+$ at a given 
time, they stay $+$ forever (they will alway keep a strict majority $+$ neighbors).

Thus if $\tilde\sigma(iL^2(\log L)^c)=+$ for all $x\in \{1,\dots, L\}^3\times \{1,\dots, i\}$, then during the time window 
$[iL^2 (\log L)^c,(i+1)L^2 (\log L)^c]$, only the spins in the layer  $\{1,\dots, L\}^3\times \{ i+1\}$, may change sign
(those above are blocked by definition). Thus for the same reason as above, the dynamic in this layer
is in bijection with the three-dimensional dynamics,
and thus at time $(i+1)L^2 (\log L)^c$ if $L$ is large 
(using the Markov property for the dynamics and \cite[Theorem 1]{CMST})
all the spins have flipped to $+$ with probability at least $L^{-k-1}$.
 
\begin{multline}
 \bP\left[\exists x\in \{1,\dots, L\}^3\times \{1,\dots, i+1\},\ 
\tilde\sigma_x ((i+1)L^2(\log L)^c)=-\right]\\
\le 
 \bP\left[\exists x\in \{1,\dots, L\}^3\times \{1,\dots, i\},\
\tilde\sigma_x (iL^2(\log L)^c)=-\right]+ L^{-k-1}.
\end{multline}
which combined with the induction hypothesis, ends the proof.
\end{proof}

\subsection{The geometrical construction}

The main reason why the strategy of the previous section does not give a sharp bound is that
 the evolution of the modified dynamics is very far from the predicted motion by mean curvature. The whole challenge is thus
to make the dynamics evolves in a way that is not too far from motion by mean curvature, but that we can still control.
This is the object of the following construction:

For any $r\ge 0$ let $S^3_r$ denote the three-dimensional discrete ball with radius $r$ ($r$ need not be an integer)

\begin{equation}
 S^3_r:=\left\{z=(z_1,z_2,z_3)\in \Z^3 \ \big| \ z_1^2+z_2^2+z_3^2\le r^2\right\}.
\end{equation}

We consider dynamics in the cylinder  
\begin{equation}
 \mathcal C_L:=  S^3_{8 L (\log L)^{c_2}} \times \{1,\dots,L\},
\end{equation}
with $+$ boundary condition at the bottom of the cylinder and around it and $-$ on the top of it, or formally:
\begin{itemize}
 \item $\eta_x=+$ for $x\in \partial S^3_{8 L (\log L)^{c_2}}\times \{1,\dots L\}$, and $S_L^3\times \{0\}$ where $\partial$ denotes here the boundary 
in $\Z^3$,
 \item $\eta_x=-$ for $S^3_{8 L (\log L)^{c_2}}\times \{L+1\}$.
\end{itemize}
We call this boundary condition $\eta_0$. We put $-$ spins at the top of the cylinder only because it turns out to be handy in the proof
 but this is not crucial point.

\medskip

Note that from this choice of boundary condition and from the way we prove of Theorem \ref{mainres} (see below), 
that the conclusion of Theorem \ref{mainres} 
still holds also if the boundary condition is not all $+$ on the boundary of the cube, 
but is $-$ one one side and $+$ on all the others.

\medskip

In most cases, we do not underline dependence in the boundary condition and write  $(\sigma^{\xi}_x(t))_{x\in \mathcal C_L}$ 
for the spin configuration at time $t$ starting from configuration $\xi$ with $\eta_0$ boundary condition.

For every $i,k\in \N$ one defines the sets $\mathcal S_{L,k}$ and $\mathcal C_L^{(i)}$ as follows :
\begin{equation}\label{sll}\begin{split}
\mathcal S_{L,k} &:=  S^3_{(8 L-k)(\log L)^{c_2}},\\
 \mathcal C_L^{(i)} & :=\mathcal C_L\cap \left\{z=(z_1,z_2,z_3,z_4)\in \Z^4 \ \big| 
\ (z_1,z_2,z_3)\in \mathcal S_{L,(i-2z_4+2)_+}\right\},
\end{split}\end{equation}
where $x_+=\max(x,0)$ denotes the positive part of $x$. Note that $\mathcal C_L^{(0)}=\mathcal C_L$.
The set $\mathcal C_L^{(i)}$ is a pile (along the fourth dimension) of $L$ three-dimensional balls of increasing radius (see Figure \ref{yty}).
For different subsets $\mathcal D \subset \mathcal C_L$ we may consider
\begin{equation}
 \mathcal D \cap \left(\Z^3\times \{j\}\right)
\end{equation}
the $j$-th \textit{slice} of $\mathcal D$. 
Our main result is a consequence of the following proposition, that controls the evolution of the set of minus in
 $\sigma^-(t)$.
We prove it by considering an auxiliary dynamics that roughly blocks the spins in $\mathcal C_L^{(i)}$ to $-$ up to time
 $(i-1)L (\log L)^{c_0}$ (in fact one needs to perform several coupling but this is the main idea),
\begin{proposition} \label{maintool}
 For every $k>0$, for every $i\in \{1,\dots,4L\}$ one has
\begin{equation}
 \bP\left[\exists t\in \left[iL(\log L)^{c_0}, L^3\right],\ 
\exists x\in \mathcal C_L\setminus  \mathcal C_L^{(i)},\ \sigma_x^-(t)=-\right]=O(L^{-k}).
\end{equation}
\end{proposition}

\begin{rem}\rm
 In addition to an upper bound for the mixing time, the above result gives a control on how the set of $-$ spins reduces through time:
with high probability, at time $iL(\log L)^{c_0}$ and after, the set of minus spins is included in $\mathcal C_L^{(i)}$. However,
the real evolution of the shape of the set of $-$ should be very different from the one described in Figure \ref{yty}.
Recall that at a macroscopic level, it is believed that the local drift of the interface between $+$ and $-$ is proportional to local mean curvature
(this conjecture is the heuristic support for  Lifshitz's law).
\end{rem}

We show now how to get our main result from the former statement, which we prove in the next section.

\begin{proof}[Proof of Theorem \ref{mainres} from Proposition \ref{maintool} when $d=4$]

Applying Proposition \ref{maintool} for $i=4L$, we get that for any given $k$, for $L$ large enough
and for $t=t_L:=4L^2(\log L)^{c_0}$
 \begin{equation}
 \bP\left[
\exists x\in \mathcal C_L\setminus  (S^3_{6 L (\log L)^{c_2}} \times \{1,\dots,L\}),\ \sigma_x^-(t_L)=-\right]=O( L^{-k}).
\end{equation}

 \begin{figure}[hlt]
 \begin{center}
 \leavevmode 
 \epsfxsize =14 cm
 \psfragscanon
 \psfrag{4LL}{$4L(\log L)^{c_2}$}
 \psfrag{6LL}{$6L(\log L)^{c_2}$}
 \psfrag{8LL}{$8L(\log L)^{c_2}$}
 \psfrag{i=0}{$\mathcal C_L^{(0)}$}
 \psfrag{L}{L}
 \psfrag{i=L}{$\mathcal C_L^{(L)}$}
 \psfrag{i=2L}{$\mathcal C_L^{(2L)}$}
 \psfrag{i=4L}{$\mathcal C_L^{(4L)}$}
 \epsfbox{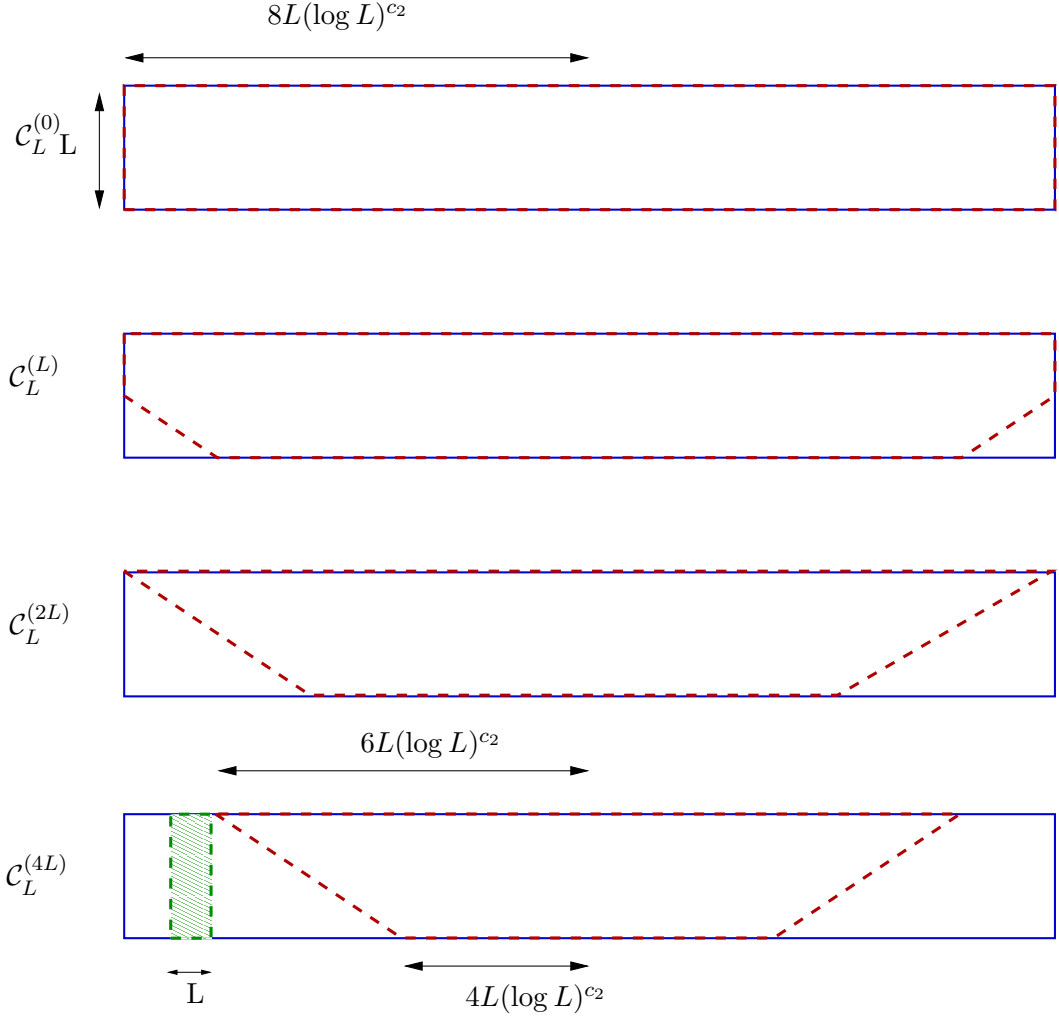}
 \end{center}
 \caption{\label{yty} Schematic two-dimensional representation of the sets $\mathcal C_L^{(i)}$ for four values of $i$. 
 The outside rectangle in full line is $\mathcal C_L$, the region within the dotted line is $\mathcal C_L^{(i)}$.
 The figure also shows how we place the hypercube $\tilde \gG_L$ (shadowed region) in 
 $\mathcal C_L^{(4L)}\setminus \mathcal C_L$ for the proof of the main Theorem. For graphical reason 
 the scaling in the different directions is not the same.}
 \end{figure}

In particular if $L$ is chosen large enough, at time $t_L$ there are no more minus spins in the cube
\begin{equation}
  \tilde\Gamma_L:= \{\lceil 6 L(\log L)^{c_2} \rceil+1,\dots,  \lceil 6 L(\log L)^{c_2}\rceil+L   \} \times \{1,\dots,L\}^3,
\end{equation}
with probability at least $1-O(L^{-k})$ since $\tilde\gG_L\subset (\mathcal C_L\setminus \mathcal C_L^{(4L)})$ ($\tilde \gG_L$ is a translation of $\gG_L$ from Theorem \ref{mainres}).
Moreover, by monotonicity in the boundary condition, the evolution of the spins in $\tilde\Gamma_L$ for the dynamics in $\mathcal C_L$ 
starting for all $-$ is stochastically dominated by the evolution
of $\tilde \Gamma_L$ started from all $-$ inside with $+$ boundary condition on $ \partial \tilde \Gamma_L$, 
and therefore \eqref{leres} holds.  
\end{proof}

\subsection{The special role of dimension three}

The proof of Proposition \ref{maintool} that we develop in the next section, 
relies heavily on the work done for the three-dimensional model.
Before starting it, we wish to discuss the particular role played by the three-dimensional model.

 In \cite{CMST}, the bound $\Tm=O(L^2 (\log L)^{c_0})$ was 
obtained by a careful analysis of the fluctuations at equilibrium and the mixing time for monotone surfaces.
A monotone surface is the graph of a function taking integer values, defined on a connected set $\gL\subset \bbZ^2$ 
and that is monotone with respect to
both coordinate in $\bbZ^2$ (this notion can be extended to any dimension). 

\medskip

The reason why 
dimension $3$ allows to get the result for higher dimension is that fluctuation of monotone surface have low amplitude 
(order $\log L$ for surfaces of side length $L$), and the reason is the following:
if we want the dynamic on the different slices of $\mathcal C_L$ not to interact between one another, one want to put the 
$+/-$ interfaces of successive slices at a distance strictly larger than the typical interface fluctuation. This is the 
reason for the particular construction of the sets $\mathcal C_L^{(i)}$.
This forces us to consider dynamics on a set $\mathcal C_L$ that has diameter $L(\log L)^{c_2}$ instead of $L$, but the 
consequence of this is just the loss of some extra $\log L$ factors in the result.

\medskip

In dimension $2$, fluctuations of monotone interfaces are too large (of order $\sqrt{L}$ for surfaces of side length $L$) 
to make a similar construction efficient  
(in particular one \textbf{could not} have proved the main result in \cite{CMST} for $d=3$ 
by using only estimates for the two dimensional model). In dimension $4$
and more, fluctuations  at equilibrium of monotone hyper-surfaces are believed to be smaller (\textit{i.e.} analogy with Gaussian Free Field, one predicts
 that the mean fluctuations are of order $1$) but
tools that we have in dimension $3$ (one-to-one mapping with dimer models, determinantal representation) 
are not available to compute them, and therefore the method of \cite{CMST} cannot be applied.

\section{Proof of Proposition \ref{maintool}}

Our proof of Proposition \ref{maintool}, which is the core of this note, uses repeatedly a variant of Proposition 3 
from \cite{CMST} that we state below. 
The result differs from the cited proposition as, although we also consider a dynamics in $S^3_r$ with $+$ boundary condition,
we put the extra condition that spins inside a sphere of radius $S^3_{r-l}$ (with $l\ge (\log r)^{c_2}/2$) are constrained to stay $-$.
Nevertheless, the proof in \cite{CMST} works without any change here.
Recall here that at the end of Section \ref{pbrsut} we fixed $c_2=3/2$, $c_1> 13/2$, and $c_0>c_1+2c_2$.

\begin{proposition}\label{uncapitano}
Given $r$ and $l\ge \frac{(\log r)^{c_2}}{2}$, consider the (three-dimensional) dynamics in $S^3_r \setminus S^3_{r-l}$ with $+$ 
boundary condition on $\partial S^3_r$ and $-$ on $\partial^-S^3_{r-l}$ (\textit{i.e.} $+$ on the outside, $-$
 on the inside).
Then for any $k>0$, uniformly in $l\ge \frac{(\log r)^{c_2}}{2}$,
\begin{equation}
 \bP\left[\exists t\in \left[\frac{1}{16}r(\log r)^{c_1},r^3\right],\ \exists x\in  S^3_r\setminus S^3_{r-1},\ \sigma^{-}_x(t)=- \ \right] = O(r^{-k}).
\end{equation}
\end{proposition}
 

Note that by monotonicity, the result has to be checked only for $l\ge (\log r)^{c_2}/2$.
From this proposition, we deduce the following result, that will be of practical use for us.

\begin{cor} \label{ndg}
 Consider the (three-dimensional) dynamics in $S^3_r \setminus S^3_{r-2(\log L)^{c_2}}$
 with  boundary condition $+$ on $\partial S^3_r$ and $-$ on $\partial^- S^3_{r-2(\log L)^{c_2}}$. Then for any $k>0$, 
for any $r\in [L(\log L)^{c_2},8L(\log L)^{c_2}]$,
\begin{equation}\label{tret}
\bP\left[\exists t\in \left[L(\log L)^{c_0},L^3\right],\ \exists x\in  S^3_r\setminus S^3_{r-(\log L)^{c_2}},\
 \sigma^{-}_x(t)=- \ \right] = O(L^{-k}).
\end{equation}
\end{cor}

\begin{proof}
Define
\begin{equation}
 \tau_i:=\inf\left\{t\ge iL (\log L)^{c_1+c_2} \ | \ \exists x \in  S^3_r\setminus S^3_{r-i}, \ \sigma^-_x(t)=-\right\},
\end{equation}
with the convention that $\inf \emptyset=\infty$.
Equation \eqref{tret} is implied by
\begin{equation}
 \bP(\tau_{\lceil(\log L)^{c_2}\rceil}\le L^3)=O(L^{-k}),
\end{equation}
because $c_0>c_1+2c_2$.

We prove by induction on $i$ that for all $i\in\{1,\dots,\lceil(\log L)^{c_2}\rceil\}$, for all large $L$,
\begin{equation}\label{troisinc}
\bP(\tau_i\le L^3)\le iL^{-(k+1)}.
\end{equation}
For $i=1$ one simply uses Proposition \ref{uncapitano} with $l=2(\log L)^{c_2}\ge \frac{(\log r)^{c_2}}{2}$.
For the induction step, one uses the induction hypothesis as follows
\begin{multline}\label{troisis}
 \bP(\tau_{i+1}\le L^3 )\le \bP(\tau_{i}\le L^3)+ \bP(\tau_{i+1}\le L^3\ ; \ \tau_i> L^3)\\
 \le iL^{-(k+1)} 
+  \bP(\tau_{i+1}\le L^3\ ; \ \tau_i> L^3).
\end{multline}
Now we need to estimate $\bP(\tau_{i+1}\le L^3\ ; \ \tau_i> L^3)$.
In order to do so we introduce an auxiliary dynamics $(\bar \sigma(t))_{t\ge 0}$, that coincides 
with $\sigma^-(t)$ up to time  
$iL(\log L)^{c_1+c_2}$, and which for $t\ge iL(\log L)^{c_1+c_2}$, uses the same clock 
process and the same coin-flips (see Remark \ref{constru}) 
as $(\sigma(t))_{t\ge i(\log L)^{c_1+c_2}}$ but cancels the moves that create a $-$ outside $S^3_{r-i}$.
Notice that up to time $\tau_i$, $\sigma^-(t)$ and $\bar \sigma(t)$ coincide and therefore 
\begin{multline} \label{ferte}
 \bP(\tau_{i+1}\le L^3\ ; \ \tau_i> L^3)\\
=\bP(\exists t\in \left[(i+1)L(\log L)^{c_1+c_2},L^3\right], \exists x \in S^3_r\setminus S^3_{r-(i+1)}, 
\ \bar \sigma_x(t)=-; \tau_i > L^3)\\
\le \bP(\exists t\in \left[(i+1)L(\log L)^{c_1+c_2},L^3\right], \exists x \in S^3_r\setminus S^3_{r-(i+1)}, 
\ \bar \sigma_x(t)=-;\\
\forall x\in S^3_r\setminus S^3_{r-i},\ \sigma^-_x(iL(\log L)^{c_2+c_1})=+), 
\end{multline}
where the last line just uses
\begin{equation}
  \{\tau_i > L^3\}\subset \{ \forall x\in S^3_r\setminus S^3_{r-i},\ \sigma^-_x(iL(\log L)^{c_2+c_1})=+\}.
\end{equation}

Now let $\big(\hat \sigma^{\xi}(t)\big)_{t\ge0}$ denote the dynamics in $S^3_{r-i}\setminus S^3_{r-2(\log L)^{c_2}}$
with  boundary condition $+$ on $\partial S^3_{r-i}$ and $-$ on $\partial^- S^3_{r-2(\log L)^{c_2}}$ started from $\xi$ and denote by $\hat \bP$ the associated probability. 
Using the Markov property at time 
$iL(\log L)^{c_1+c_2}$, one gets that the last line of \eqref{ferte} is equal to
\begin{multline}
 \bE \Bigg[ \ind_{\{\forall x\in S^3_r\setminus S^3_{r-i}, \sigma^-_x(iL(\log L)^{c_2+c_1})=+\}} 
\hat \bP\Big(\exists t\in [L(\log L)^{c_1+c_2}, L^3-iL (\log L)^{c_1+c_2}]\\
\exists x\in S^3_{r-i}\setminus S^3_{r-(i+1)},\
 \hat \sigma^{\sigma^-(iL(\log L)^{c_1+c_2})}_x(t)=-\Big)\Bigg]\\
\le \hat \bP\Big(\exists t\in [L(\log L)^{c_1+c_2}, L^3-iL (\log L)^{c_1+c_2}]\
\exists x\in S^3_{r-i}\setminus S^3_{r-(i+1)},\
 \hat \sigma^{-}_x(t)=-\Big)
\end{multline}
where in the first line $\sigma^-(iL(\log L)^{c_1+c_2})$ is, with some abuse of notation, considered as an element 
 of $\gO_{S^3_{r-i}\setminus S^3_{r-2(\log L)^{c_2}}}$. The second line is obtained by monotonicity.
\medskip

Note that, provided that $L$ is large enough,  for every values of $r$ and $i$ that we consider,
\begin{equation*}
  \left((r-i)-(r-2(\log L)^{c_2})\right)\ge \frac{1}{2}(\log (r-i))^{c_2}
\end{equation*}
 and  
\begin{equation*}
 L(\log L)^{c_1+c_2}\ge \frac{1}{16}r(\log r)^{c_1}
\end{equation*}
so that one can use Proposition \ref{uncapitano} and get
\begin{equation}
 \hat \bP\left(
\exists t\in [L(\log L)^{c_1+c_2}, L^3-iL (\log L)^{c_1+c_2}/2]\ \exists x\in S^3_{r-i}\setminus S^3_{r-(i+1)},
 \hat \sigma^{-}_x(t)=-\right)\le L^{-k-1},
\end{equation}
and therefore
\begin{equation}\label{crocrco}
  \bP(\tau_{i+1}\le L^3\ ; \ \tau_i> L^3)\le L^{-k-1}.
\end{equation}
Plugging \eqref{crocrco} into \eqref{troisis} gives \eqref{troisinc} for $i+1$ and
ends the induction step.

%
%

\end{proof}

We are now ready to prove Proposition \ref{maintool}. The main idea in the proof is to control the dynamics by cutting 
$\mathcal C_L$ into $L$ three-dimensional slices and to control the evolution of each slice by using the results we have for the dynamics in 
three-dimensions. However this stochastic domination is not
straightforward and comes from the \textit{ad hoc} construction of the sets $\mathcal C_L^{(i)}$.

\begin{rem}\rm
The geometric strategy we use to control the evolution of the set of $-$ in $\mathcal C_L$ presents some similarities with the one 
used by Caputo, Martinelli and Toninelli to prove Theorem 4.1 in \cite{CMT} concerning the mixing time of the dynamics of volume-biased plane 
partitions. 
\end{rem}

Recall that we consider the dynamics in the cylinder  
$\mathcal C_L$
with boundary condition $\eta_0$ described in the introduction ($+$ everywhere except at the top).
Define

\begin{equation}
 T_i:=\inf\left\{ t\ge iL(\log L)^{c_0} \ | \ 
\exists x\in \mathcal C_L\setminus  \mathcal C_L^{(i)},\ \sigma_x^-(t)=-\right\}.
\end{equation}
Proposition \ref{maintool} can be expressed as: for all $i\in\{1, \dots, 4L\}$, for every $k$ and for sufficiently large $L$ (how large depending on $k$)
\begin{equation}
 \bP(T_i\le L^3)\le iL^{-(k+1)}. 
\end{equation}
We prove it by induction on $i$.

 We start with the case $i=1$. In that case, one just has to show
that the set of $-$ in the first slice of the cylinder $\mathcal C_L$ has decreased, \textit{i.e.} that (recall notation \eqref{sll})
\begin{multline}
 \bP\big(
\exists t\in [L(\log L)^{c_0}, L^3],
 \exists x\in \big( \mathcal S_{L,0}\setminus \mathcal S_{L,1}\big)\times \{1\},\
 \sigma^{-}_x(t)=-\big)\le L^{-(k+1)}.
\end{multline}

By monotonicity in the boundary condition, the evolution of the spins on 
$\mathcal S_{L,0}\times\{1\}$ for the dynamics on $\mathcal C_L$ with
$\eta_0$ boundary condition dominates the evolution of the spins for the dynamics on
$\mathcal S_{L,0}\times\{1\}$ with $+$ boundary condition on $(\mathcal S_{L,0}\times\{0\} )
\cup (\partial \mathcal S_{L,0}\times\{1\})$ and $-$ on $\mathcal S_{L,0}\times\{2\}$. 

\medskip

One can then check that the latter dynamics corresponds
(up to a trivial bijection) to the (three-dimensional) dynamics on $\mathcal S_{L,0}$ with $+$ boundary condition.
Indeed, for each $x\in \mathcal S_{L,0}\times\{1\}$, due to boundary condition one has $\sigma_{x+e_4}(t)=-$, $\sigma_{x-e_4}(t)=+$ 
(where $e_4:=(0,0,0,1)$), so that
the effects on the dynamics of the neighbors in the fourth direction above and below cancel out.
Once one has noticed this, the case $i=1$ is an immediate consequence of Corollary \ref{ndg} (without the restriction of
 having spin blocked inside the inner-sphere).

%

\medskip

For the induction step, we use a strategy similar to the one used in the proof of Corollary \ref{ndg}. First, notice that
\begin{multline}
 \bP(T_{i+1}\le L^3)\le \bP(T_i\le L^3)+\bP(T_{i+1}\le L^3\ ; T_i> L^3)\\
\le iL^{-(k+1)}+\bP(T_{i+1}\le L^3\ ; T_i> L^3),
\end{multline}
where we used the induction hypothesis in the last inequality.
Now we need to estimate $\bP(T_{i+1}\le L^3 \ ; \ T_i> L^3)$.
In order to do so we introduce an auxiliary dynamics $(\bar \sigma(t))_{t\ge 0}$, that coincides with $(\sigma^-(t))_{t\ge 0}$ up to time $
iL(\log L)^{c_0}$ and which for $t\ge iL(\log L)^{c_0}$ uses the same clock 
process and the same coin-flips (see Remark \ref{constru}) 
as $(\sigma(t))_{t\ge iL(\log L)^{c_0}} $ but cancels the moves that create a $-$ outside of $\mathcal C_L^{(i)}$.
Note that $\bar \sigma(t)$ and $\sigma(t)$ coincide up to time $T_{i}$.
Therefore

\begin{multline} \label{ferte2}
 \bP(T_{i+1}\le L^3\ ; \ T_i> L^3)\\
=\bP(\exists t\in \left[(i+1)L(\log L)^{c_0},L^3\right] \exists x \in \mathcal C_L\setminus \mathcal C^{(i+1)}_L, \ 
\bar \sigma_x(t)=-\ ; \ T_i > L^3)\\
\le \bP(\exists t\in \left[(i+1)L(\log L)^{c_0},L^3\right] \exists x \in \mathcal C_L\setminus \mathcal C^{(i+1)}_L, \ \bar \sigma_x(t)=-\ ;\\
\forall x\in \mathcal C_L\setminus \mathcal C^{(i)}_L, \sigma^-_x(iL(\log L)^{c_0})=+).
\end{multline}
Now let $\big(\hat \sigma^{-}(t) \big)_{t\ge0}$ denote the dynamics in $\cC_L^{(i)}$ with $-$ boundary condition on 
$(S^3_{8L(\log L)^{c_2}}\times \{L+1\})\cap \partial \cC_L^{(i)}$ (the top)
and $+$ everywhere else.  As in the proof of Corollary \ref{ndg}, by the Markov property at time $iL(\log L)^{c_0}$  and monotonicity
the last line of \eqref{ferte2} is at most

\begin{equation}\label{ferte3}
 \bP\left[\exists t \in \left[L(\log L)^{c_0},L^3-iL(\log L)^{c_0}\right]
 \exists x \in \cC_L^{(i)}\setminus \cC_L^{(i+1)},\ \hat \sigma^{-}_x(t)=-\right].
\end{equation}

Now notice, that with our particular construction  of $\cC_L^{(i)}$ (see Figure \ref{fg}) the 
$j$-th slice of $\cC_L^{(i)}$ is equal to the $j+1$-th slice of $\cC_L^{(i+2)}$ for all $j\le \lceil i/2\rceil$, so that
 for every $x\in\cC_L^{(i)}\setminus \cC_L^{(i+2)}$
one has
$x\pm e_4 \notin \cC_L^{(i)}\setminus \cC_L^{(i+2)}$.

\medskip

More precisely one has for every $i\in\{0,\dots, 2L-2\}$
\begin{multline}\label{bdecop}
 \partial\left(\cC_L^{(i)}\setminus \cC_L^{(i+2)}\right)=(\cC_L^{(i)}\setminus \cC_L^{(i+2)}-e_4)\sqcup(\cC_L^{(i)}\setminus \cC_L^{(i+2)}+e_4)\\
\sqcup\left((\partial^- \mathcal S_{L,(i+2)(\log L)^{c_2}} \times \{1\}\right) \sqcup \left( (\partial \mathcal S_{L,0})\times \{\lceil i/2 \rceil+1\} \right).
\end{multline}
The same holds for $i\in\{2L-1, \dots 4L\}$, but with $(\partial \mathcal S_{L,0})\times \{\lceil i/2 \rceil+1\}$
replaced by $(\partial \mathcal S_{L,(i-2(L-1))})\times \{L\}$. Here $\sqcup$ denotes the disjoint union.

\medskip

Let $(\tilde \sigma(t))_{t\ge 0}$ denote the dynamics in $\cC_L^{(i)}\setminus \cC_L^{(i+2)}$ with the following boundary condition $\eta^{(i)}$:
\begin{itemize}
 \item $-$ on $(\cC_L^{(i)}\setminus \cC_L^{(i+2)})+e_4$ and on $(\partial^- \mathcal S_{L,(i+2)}) \times \{1\}$, 
 \item $+$ on $(\cC_L^{(i)}\setminus \cC_L^{(i+2)})-e_4$ and on $(\partial \mathcal S_{L,0})\times \{\lceil i/2 \rceil+1\}$ or $(\mathcal S_{L,(i-2(L-1))})\times \{L\}$
\end{itemize}
(see Figure \ref{fg} for a graphical description of the boundary conditions).
Once again, monotonicity in the boundary condition gives that \eqref{ferte3} is at most

\begin{equation}
 \bP\left[\exists t \in \left[L(\log L)^{c_0},L^3\right] 
\exists x \in \cC_L^{(i)}\setminus \cC_L^{(i+1)},\ \tilde \sigma^{-}_x(t)=-\right].
\end{equation}
To finish the proof we just need to prove that the above quantity is less than $L^{-k-1}$.

\medskip

The key point is to notice that the boundary of $\cC_L^{(i)}\setminus \cC_L^{(i+2)}$ is the union of the boundary of its slices
, i.e.\
\begin{equation}
  \partial\left(\cC_L^{(i)}\setminus \cC_L^{(i+2)}\right)=\bigcup_{j=1}^L
\partial \left[\left(\mathcal S_{L,(i-2(j-1))_+}\setminus \mathcal S_{L,(i-2(j-1)+2)_+}\right)\times \{j\}\right].
\end{equation}
Therefore
the evolution of the spins in the different slices are independent.
Moreover, due to our choice of boundary condition, the evolution of
$(\tilde \sigma(t))_{t\ge 0}$ in the $j$-th slice of $\mathcal C_L^{(i)}\setminus \mathcal C_L^{(i+2)}$
 
\begin{center}
 $\left(\mathcal S_{L,(i-2(j-1))_+}\setminus \mathcal S_{L,(i-2(j-1)+2)_+}\right)\times \{j\}$
\end{center}
 is the same (up to a trivial bijection) as the evolution of the spins for the dynamics in 
\begin{center}
 $\mathcal S_{L,(i-2(j-1))_+}\setminus \mathcal S_{L,(i-2(j-1)+2)_+}$
\end{center}
with $+$ boundary conditions on $\partial \mathcal S_{L,(i-2(j-1))_+}$ and  $-$ on  $\partial^{-} S_{L,(i-2(j-1)+2)_+}$.

\medskip

The reason for that is that for every $x\in\left(\mathcal S_{L,(i-2(j-1))_+}\setminus \mathcal S_{L,(i-2(j-1)+2)_+}\right)\times \{j\}$
the boundary condition we have chosen imposes that $\sigma_{x-e_4}=+$ and $\sigma_{x+e_4}=-$ and therefore, 
influence of neighbors along the fourth direction cancels out.
For the rest of the boundary, $\eta^{(i)}$ imposes
 $+$ boundary conditions on $(\partial \mathcal S_{L,(i-2(j-1))_+})\times \{j\}$ and $-$ on  $(\partial^{-} S_{L,(i-2(j-1)+2)_+})\times \{j\}$.
This appears well on figure \ref{fg}.

 \begin{figure}[hlt]
 \begin{center}
 \leavevmode 
 \epsfxsize =14 cm
 \psfragscanon
 \psfrag{CLIP2}{$\mathcal C_L^{(i+2)}$}
 \psfrag{Z}{$\mathbb Z$}
 \psfrag{Z3}{$\mathbb Z^3$}
 \epsfbox{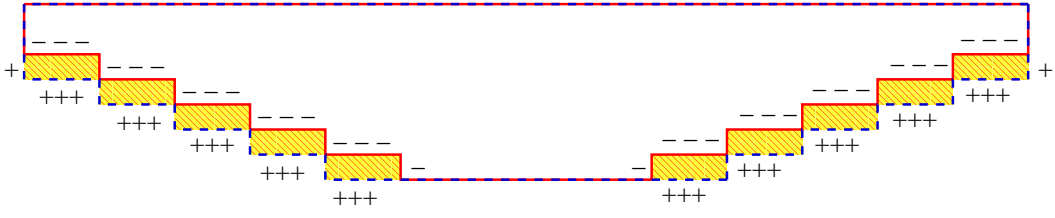}
 \end{center}
 \caption{ \label{fg} Schematic two dimensional view of the boundary condition considered for the dynamics 
$\big(\tilde \sigma(t)\big)_{t\ge0}$ (for $i\le 2L-2$, $i$ even). 
 Here the full line denotes the border of $\mathcal C_L^{(i+2)}$, the thick dotted line denotes the border of 
  $\mathcal C_L^{(i)}$  and the shadowed region is $\mathcal C_L^{(i)}\setminus \mathcal C_L^{(i+2)}$.
 The boundary conditions are also represented: for every site in
 $\mathcal C_L^{(i)}\setminus \mathcal C_L^{(i+2)}$, the neighbors along the fourth direction are outside of 
 $\mathcal C_L^{(i)}\setminus \mathcal C_L^{(i+2)}$ and the two neighbors along that direction are given
  opposite spins by the boundary conditions.}
 \end{figure}

Hence, using again Corollary \ref{ndg}, one gets that for $L$ large enough, for every  $i\in[1,4L]$, $j\in [1,L]$, one has

\begin{multline}
  \bP\big[\exists t\in \left[L(\log L)^{c_0}, L^3\right]
\exists x\in \left(\mathcal S_{L,(i-2(j-1))_+}\setminus \mathcal S_{L,(i-2(j-1)+1)_+}\right)\times \{j\} ,\\
 \tilde \sigma_x^{-}(t)=-\big]\le L^{-(k+2)}.
\end{multline}
Using the union bound and summing over $j\in \{1,\dots,L\}$ one gets that 
\begin{multline}
 \bP(T_{i+1}\le L^3\ ;\ T_{i}> L^3)\\
\le  \bP\left[\exists t \in \left[L(\log L)^{c_0},L^3\right] 
\exists x \in \cC_L^{(i)}\setminus \cC_L^{(i+1)},\ \tilde \sigma^{-}_x(t)=-\right]\le L^{-(k+1)},
\end{multline}
which finishes the proof.
\qed

\section{From $4$ dimensions to $d$ dimensions}

We briefly indicate how the proof should be modified to work also in higher dimensions.
One replaces $\mathcal C_L$ by

\begin{equation}
 \mathcal C_{L,d}:=  S^3_{2 d L (\log L)^{c_2}} \times \{1,\dots,L\}^{d-3}.
\end{equation}
One considers the dynamics in $\cC_{L,d}$ with boundary condition
\begin{itemize}
 \item  $+$ on $(\partial S^3_{2 d L (\log L)^{c_2}})\times \{1,\dots,L\}^{d-3}$ and on 
$\bigcup_{i=1}^{d-3} S^3_{2 d L (\log L)^{c_2}}\times F_i$ 
where $F_i:=\{1,\dots,L\}^{i-1}\times\{0\}\times \{1,\dots,L\}^{d-3-i}$.
 \item $-$ on $\bigcup_{i=1}^{d-3} S^3_{2 d L (\log L)^{c_2}}\times F'_i$ where $F'_i:=\{1,\dots,L\}^{i-1}\times\{L+1\}\times \{1,\dots,L\}^{d-3-i}$.
\end{itemize}
 
Define:

\begin{equation}\begin{split}
 \mathcal S_{L,d}^{(i,k)}&:=S^{3}_{2d L \log L^{c_2}-(i-k+2(d-3))_+(\log L)^{c_2}},\\
 \mathcal C_{L,d}^{(i)}&:=\mathcal C_{L,d} \cap \left\{z=(z_1,z_2,z_3,z_{(d-3)})\in \Z^d \ \big| 
\ (z_1,z_2,z_3)\in \mathcal S^{(i,|z^{(d-3)}|_1)}_{L,d}\right\}
\end{split}
\end{equation}

where $z^{(d-3)}=(z_4,\dots,z_d)\in \bbZ^{d-3}$ and $|\dots|_1$ denotes the $l_1$ norm on $\bbZ^{d-3}$
\begin{equation}
|z^{(d-3)}|_1:= \sum_{j=4}^d |z_j|.
\end{equation}

Then one can prove the following generalization of Proposition \ref{maintool}
\begin{proposition} \label{maintool2}
 For every $k>0$, for every $i\in \{1,\dots, 2(d-2)L \}$ one has
\begin{equation}
 \bP\left[\exists t\in \left[iL(\log L)^{c_0}, L^3\right]\ 
\exists x\in \mathcal C_{L,d}\setminus  \mathcal C_{L,d}^{(i)},\ \sigma_x^-(t)=-\right]=O(L^{-k}).
\end{equation}
\end{proposition}

The proof is performed in the same way, using induction on $i$. Note that the proof is still based on 
the three-dimensional result Corollary \ref{ndg} and not on an higher dimensional version of it (that we would be unable to prove). 
Coupling, stochastic comparison and identification with 
the three-dimensional dynamics are valid in that case also.
The fact that on any intermediate result, the value of $k$ (when one shows a quantity is less than $L^{-k}$) can be chosen arbitrarily large, 
allows to perform a union bound in any
 dimension without any harm.
In the induction step, using coupling and stochastic comparison, one is left to study the dynamics in 
$\mathcal C_{L,d}^{(i)}\setminus \mathcal C_{L,d}^{(i+2)}$ with specific boundary condition.

\medskip

Note that in analogy with \eqref{bdecop} the boundary of 
$\partial(\mathcal C_{L,d}^{(i)}\setminus \mathcal C_{L,d}^{(i+2)})$ can be decomposed in four parts as follows 

\begin{multline}
\partial(\mathcal C_{L,d}^{(i)}\setminus \mathcal C_{L,d}^{(i+2)})=\left[\bigsqcup_{j=4}^{d} (\mathcal C_{L,d}^{(i)}\setminus \mathcal C_{L,d}^{(i+2)}+e_j)\right]
\sqcup\left[\bigsqcup_{j=4}^{d} (\mathcal C_{L,d}^{(i)}\setminus \mathcal C_{L,d}^{(i+2)}-e_j)\right]\\
\sqcup  \left((\partial^- \mathcal S_{L,d}^{i+2,|{\bf 1}^{(d-3)}|_1})\times \{{\bf 1}^{(d-3)}\}\right)\sqcup \mathcal K,  
\end{multline}
where ${\bf 1}^{(d-3)}=(1,\dots,1)\in \Z^{d-3}$ and $\mathcal K$ denotes what remains of the boundary 
when the three other parts have been taken away. 

\medskip

Our specific boundary condition is

\begin{itemize}
 \item $+$ on $\bigsqcup_{j=4}^{d} (\mathcal C_{L,d}^{(i)}\setminus \mathcal C_{L,d}^{(i+2)}-e_j)$ and $\mathcal K$,
  \item $-$ on $\bigsqcup_{j=4}^{d} (\mathcal C_{L,d}^{(i)}\setminus \mathcal C_{L,d}^{(i+2)}-e_j)$  and $\left(\partial^- S_{L,d}^{i+2,|{\bf 1}^{(d-3)}|_1}\times \{{\bf 1}^{(d-3)}\}\right)$.
\end{itemize}
One can check that with these boundary conditions, the different slices
\begin{center}
  $\left(\mathcal S^{(i,|z_{(d-3)}|_1)}_{L,d}\setminus \mathcal S^{(i+2,|z_{(d-3)}|_1)}_{L,d}\right)\times \{z^{(d-3)}\}$ 
\end{center}
of the system evolve independently for each $z^{(d-3)}\in\{1,\dots,L\}^{d-3}$   and have the same evolution as
the corresponding three-dimensional dynamics with appropriate boundary condition (for $j\in \{4, \dots, d\}$ 
for every  $x\in \mathcal C_{L,d}^{(i)}\setminus \mathcal C_{L,d}^{(i+2)}$, $x\pm e_j$  belong to the boundary 
and our boundary condition ensures that the influence of spins of $x\pm e_j$ cancels out). One can apply 
Corollary \ref{ndg} and perform a union bound on $z^{(d-3)}$. Details are omitted.

\medskip

 {\bf Acknowledgements:} 
 The author is very grateful to P.\ Caputo, F.\ Martinelli and F.\ Toninelli for many useful scientific 
 discussions and precious help on the manuscript. This work was written during the authors postdoctoral  stay in Universita di Roma Tre 
 supported by European Research Council grant PETRELSS 228032.
He acknowledges hospitality and support.

\end{document}